\documentclass[10pt,twocolumn]{IEEEtran}

\usepackage{amsmath,amsthm,epsfig,amssymb,verbatim,amsopn,cite,subfigure,multirow}
\usepackage{balance,nopageno}
\usepackage[usenames,dvipsnames]{color}
\usepackage[all]{xy}  
\usepackage{url}
\usepackage{amsfonts}
\usepackage{amssymb}
\usepackage{epsfig}
\usepackage{epstopdf}
\usepackage{slashbox}
\usepackage{bm}
\usepackage{cases}

  {\proof}{\proofend}
\newtheorem{proposition}{Proposition}

\DeclareMathOperator*{\argmax}{arg\,max}

\newcommand{\ve}[1]{\boldsymbol{#1}}

\newcommand{\vH}{\ve{H}} \newcommand{\vh}{\ve{h}}

\newcommand{\be}{\begin{equation}} \newcommand{\ee}{\end{equation}}
\newcommand{\bea}{\begin{eqnarray}} \newcommand{\eea}{\end{eqnarray}}

\newcommand{\qh}{{\bf h}}

\newcommand{\qn}{{\bf n}}

\newcommand{\qr}{{\bf r}}

\newcommand{\qw}{{\bf w}}
\newcommand{\qx}{{\bf x}}

\newcommand{\qB}{{\bf B}}

\newcommand{\qD}{{\bf D}}

\newcommand{\qH}{{\bf H}}
\newcommand{\qI}{{\bf I}}

\newcommand{\qW}{{\bf W}}

\newcommand{\hRD}{\vh_{RD}}
\newcommand{\hSR}{\vh_{SR}}

\newcommand{\thSR}{\tilde{\vh}_{SR}}
\newcommand{\thRD}{\tilde{\vh}_{RD}}

\newcommand{\MRT}{\mathsf{MRT}}
\newcommand{\MRC}{\mathsf{MRC}}
\newcommand{\TZF}{\mathsf{TZF}}
\newcommand{\RZF}{\mathsf{RZF}}
\newcommand{\ZF}{\mathsf{ZF}}
\newcommand{\HD}{\mathsf{HD}}
\newcommand{\gamth}{\gamma_{\mathsf{th}}}

\newcommand{\Pout}{\mathsf{P_{out}}}

\newcommand{\diag}{\mathsf{diag}}
\newcommand{\Bta}{\mathsf{Beta}}

\newcommand{\Srd}{d_2^{\tau}}
\newcommand{\Ssr}{d_1^{\tau}}

\newcommand{\bGsd}{ \frac{d_1^{\tau} z}{ \rho_1}}

\newcommand{\Sap}{\sigma_{RR}^2}

\newcommand{\Prob}{\textnormal{Pr}}

\newcommand{\AuthorOne}{Mohammadali Mohammadi$^\dag$}
\newcommand{\AuthorTwo}{Himal A. Suraweera$^\ddag$}
\newcommand{\AuthorThree}{Gan Zheng$^\S$}
\newcommand{\AuthorFour}{Caijun Zhong$^*$}
\newcommand{\AuthorFive}{Ioannis Krikidis$^\Im$}

\usepackage{multibib}
\usepackage[nodisplayskipstretch]{setspace}
\newcites{Prim}{Very important papers}

\definecolor{light-gray}{gray}{0.65}

\newcounter{mytempeqcounter}

\newcommand{\ThankFour}{The work of C. Zhong was partially supported by the Zhejiang Provincial Natural Science Foundation of China (LR15F010001) and the Fundamental Research Funds for Central Universities (2014QNA5019).

The work of I. Krikidis was supported by the Research Promotion Foundation, Cyprus under the project KOYLTOYRA/BP-NE/0613/04 ``Full-Duplex Radio: Modeling, Analysis and Design (FD-RD)".}

\title{Full-Duplex MIMO Relaying Powered by Wireless Energy Transfer}
\author{\authorblockN{\AuthorOne,\:\AuthorTwo,\: \AuthorThree,\:\AuthorFour,\:and \AuthorFive,}\\
\footnotesize{
$^\dag$Faculty of  Engineering, Shahrekord University, Iran (e-mail: m.a.mohammadi@eng.sku.ac.ir)\\
$^\ddag$Department of Electrical and Electronic Engineering, University of Peradeniya, Sri Lanka (e-mail: himal@ee.pdn.ac.lk)\\
$^*$School of Computer Science and Electronic Engineering, University of Essex, UK  (e-mail: ganzheng@essex.ac.uk)\\
$^\S$Institute of Information and Communication Engineering, Zhejiang University, China (e-mail: caijunzhong@zju.edu.cn)\\
$^\Im$Department of Electrical and Computer Engineering, University of Cyprus, Cyprus (e-mail: krikidis@ucy.ac.cy)
}}\normalsize

\begin{document}

\maketitle
\thispagestyle{empty}

\begin{abstract}
We consider a full-duplex decode-and-forward system, where the wirelessly powered relay employs the time-switching protocol to receive power from the source and then transmit information to the destination. It is assumed that the relay node is equipped with two sets of antennas to enable full-duplex communications. Three different interference mitigation schemes are studied, namely, 1) optimal 2) zero-forcing and 3) maximum ratio combining/maximum ratio transmission. We develop new outage probability expressions to investigate delay-constrained transmission throughput of these schemes. Our analysis show interesting performance comparisons of the considered precoding schemes for different system and link parameters.
\let\thefootnote\relax\footnotetext{\ThankFour}
\end{abstract}
\vspace{-0.5em}
\section{Introduction}\label{GENERAL_Introduction}
Most wireless radios so far have adopted half-duplex (HD) communications due to the challenge of handling loopback interference (LI) generated from simultaneous transmit/receive operation. However, thanks to the progress made on LI suppression recently, full-duplex (FD) communications have emerged as a viable option~\cite{Sabharwal:JSac:2014,Riihonen:TSP:2011,Riihonen:TWC:2011}. In theory, FD operation can double the HD capacity, hence is a key enabling technique for 5G systems.

On the other hand energy harvesting communications is a new paradigm that can power wireless devices by scavenging energy from external resources such as solar, wind, ambient RF power etc. Energy harvesting from such sources are not without challenges due the unpredictable nature of these energy sources. To this end, wireless energy transfer has been proposed as a promising technique for a variety of wireless networking applications \cite{ZDing:CMag:2015}.

RF signals can carry both information and energy and pioneering contributions quantifying this fundamental tradeoff have been reported. In order to remedy practical issues (same signal can not be used for both decoding and rectifying) associated with simultaneous information and energy transfer, two practical techniques, i.e., time-switching (TS) and power-splitting (PS) were proposed in \cite{Nasir:TWC:2013}. Both TS and PS apply in different network topologies and integration of the RF energy transfer into cooperative relay networks is an interesting research topic. Different HD relay networks have been studied considering amplify-and-forward (AF) relaying~\cite{Nasir:TWC:2013,Chen:TSP:2015}, relay selection~\cite{Himal:JSAC:2015} and multiple antenna relay systems~\cite{Krikidis:TCOM:2015}. In~\cite{Caijun:TCOM:2014}, the achievable throughput of FD AF and DF relaying systems with TS has been studied. However,~\cite{Caijun:TCOM:2014} only assumed single transmit/receive antenna at the relay. 

Inspired by the FD approach, in this paper we consider a source-relay-destination scenario where the multiple antenna FD relay is powered via wireless energy transfer from the source. The reason for the adoption of multiple antennas at the relay is two-fold: (1) employment of an antenna array helps the relay to accumulate more energy (2) spatial LI cancellation techniques can be deployed. Specifically, we investigate the outage probability and the delay constrained throughput by considering several precoding schemes at the relay. Our results are general in the sense that we consider arbitrary number of receive/transmit antennas at the FD relay input/output.

In summary, the contributions of this work are as follows:
\begin{enumerate}
\item
Assuming different precoding schemes at the relay, namely, optimal, zero-forcing (ZF) and maximum ratio combining (MRC)/maximum ratio transmission (MRT), we develop new expressions for the system's outage probability, which are helpful to investigate the effect of key system parameters on performance metrics such as the outage probability and delay-constrained throughput.
\item
In the case of ZF, we present simple high signal-to-noise ratio (SNR) expressions for the outage probability which enable the characterization of the system's diversity order and array gain. Moreover, we compare the performance of FD and HD modes to show the benefits of FD operation.
\end{enumerate}
\vspace{-0.8em}
\section{System Model}\label{GENERAL_system_model}
We consider a DF relaying network consisting of one
source $S$, one relay $R$, and one destination $D$. Both $S$ and $D$ are equipped with a single antenna, while $R$ is equipped with $M_R$ receive (input) antennas and $M_T$ transmit (output) antennas to enable FD operation. We assume that the $S$ to $D$ link does not exist, due to severe shadowing and path loss effect.

It is also assumed that the relay has no external power supply, and is powered through wireless energy transfer from the source. We adopt the time-sharing protocol~\cite{Nasir:TWC:2013}, hence the entire communication process is divided into two phases, i.e., for a transmission block time $T$,
$\alpha$ fraction of the block time is devoted for energy harvesting and the remaining time, $(1-\alpha)T$, is used for information transmission. It is also assumed that the channels experience Rayleigh fading and remain constant over the block time $T$ and varies independently and identically from one block to the other.

During the energy harvesting phase, the received signal $\qr_e$ at the relay can be expressed as
\vspace{-0.6em}
\begin{align}
  \qr_e= \sqrt{\frac{P_S}{d_1^{\tau}}} \qh_{SR}x_e + \qn_R,
\end{align}
where $P_S$ is the source transmit power, $d_1$ is the distance between
the source and relay, $\tau$ is the path loss exponent, $\qh_{SR}$ is the $M_R \times 1$ channel vector for the $S$-$R$ link, i.e., input antennas at $R$ are connected to the rectifying antenna (rectenna), $x_e$ is the energy symbol with unit energy, and $\qn_R$ is the zero mean additive white Gaussian noise (AWGN) with unit variance. We assume that the energy collected during the first phase is fully consumed by the relay to forward the source signal to the destination. Hence, the relay transmit power can be computed as $P_r = \frac{\kappa}{d_1^{\tau}} P_S \|\qh_{SR}\|^2$,
where $\kappa\triangleq\frac{\eta\alpha }{1-\alpha}$ with $\eta$ denoting the energy conversion efficiency.

Now, let us consider the information transmission phase. The received signal at $R$ can be written as
\vspace{-0.7em}
  \be\label{eqn:rn}
    \qr[n] = \sqrt{\frac{P_S}{d_1^{\tau}}} \qh_{SR}x_S[n] + \qH_{RR} \qx_R[n] + \qn_R[n],
  \ee
where $x_S[n]$ is the source information symbol with unit energy, and $\qx_R[n]$ is the transmitted relay signal satisfying ${\tt E}\left\{\qx_R[n]\qx^{\dag}_R[n]\right\}=P_r$, and $\qH_{RR}$ denotes the $M_R \times M_T$ LI channel. Upon receiving the signal, $R$ first applies a $ 1 \times M_R$ linear combining vector $\qw_r$ on $\qr[n]$ to obtain an estimate of $x_S$, then forwards signal to the destination $D$ using the $M_T \times 1$ transmit beamforming vector $\qw_t$. It is assumed that $\|\qw_t\|=\|\qw_r\|=1$.

The relay's estimate $\hat x_S[n]=\qw_r \qr[n]$ can be expressed as
\vspace{-0.7em}
  \begin{align}
   \hat x_S[n] &\!=\!  \sqrt{\frac{P_S}{\Ssr}} \qw_r\qh_{SR}x_S[n] \!+\! \qw_r\qH_{RR} \qx_R[n]\! +\!
   \qw_r\qn_R[n].
  \end{align}

The relay transmit signal is given by $\qx_R[n] = \sqrt{P_r}\qw_t x_S[n-\delta]$ where $\delta$ accounts for the time delay caused by relay processing. Finally, the received signal at $D$ is
\vspace{-0.6em}
\be\label{eqn:ydn}
    y_D[n] = \sqrt{\frac{1}{\Srd}}\qh_{RD}\qx_R[n] + n_D[n].
  \ee
  \vspace{-0.1em}
where $\qh_{RD}$ is the $1 \times M_T$ channel vector of the $R-D$ link, $d_2$ is the distance between the relay and destination, and $n_D$ is the zero mean AWGN with unit variance.

With the DF protocol, end-to-end signal-to-interference-plus-noise ratio (SINR) can be written as
\vspace{-0.7em}
\begin{align}
& \gamma\!= \!\label{eq: e2e snr general}\\
 & \min\!\left(\! \frac{P_S|\qw_r \qh_{SR}|^2}{ \kappa P_S
\|\qh_{SR}\|^2 |\qw_r \qH_{RR}\qw_t|^2 \!+\!
\Ssr},\!\frac{\kappa P_S}{\Ssr\Srd} \|\qh_{SR}\|^2|
\qh_{RD}\qw_t|^2
    \!\right)\!\!.\notag
\end{align}
\vspace{-0.9em}
\section{Joint Precoding/Decoding Designs}
In this section, we consider several precoder/decoding designs to suppress/cancel the effect of LI at the relay, each of which offers different performance-complexity tradeoff.
\vspace{-0.6em}
\subsection{The Optimal Scheme}
In this subsection, our main objective is to jointly design the precoder and the decoder at the FD relay so that the end-to-end SINR in~\eqref{eq: e2e snr general} is
maximized. Specifically, for a fixed value of $\alpha$, the SINR maximization problem can be formulated as
\vspace{-0.5em}
 \bea\label{eqn:opt}
    \max_{\qw_t, \qw_r} &&  \gamma~(\text{in Eq.~\eqref{eq: e2e snr general}})\\
    \mbox{s.t.} && \|\qw_r\|=\|\qw_t\|=1\notag.
 \eea
In order to solve the problem in~\eqref{eqn:opt}, we first fix $\qw_t$ and optimize $\qw_r$ to maximize $\gamma$. Therefore, the optimization problem can be re-formulated as
\vspace{-0.5em}
\bea\label{eqn:opt:Rayleigh ratio problem}
   \max_{\qw_r} &&  \frac{|\qw_r \qh_{SR}|^2}{ \frac{\kappa P_S}{\Ssr}
\|\qh_{SR}\|^2 |\qw_r \qH_{RR}\qw_t|^2+1},\\
    \mbox{s.t.} && \|\qw_r\|=1,\notag
\eea
which is a generalized Rayleigh ratio problem. It is well known that the objective function in~\eqref{eqn:opt:Rayleigh ratio problem} is globally maximized when 
\vspace{-0.5em}
 \be
    \qw_r = \frac{\qh_{SR}^{\dag}\left( \frac{\kappa P_S}{\Ssr} \|\qh_{SR}\|^2
\qH_{RR}\qw_t\qw_t^\dag\qH_{RR}^\dag + \qI \right)^{-1}
}{\left\|\qh_{SR}^{\dag}\left(  \frac{\kappa P_S}{\Ssr} \|\qh_{SR}\|^2 \qH_{RR}\qw_t\qw_t^\dag\qH_{RR}^\dag + \qI \right)^{-1}
\right\|}.
 \ee

Accordingly, by substituting $\qw_r$ into the objective function, the optimization problem in~\eqref{eqn:opt} can be re-expressed as
\vspace{-0.4em}
  \bea\label{eqn:opt:with opt w_r}
    \max_{\qw_t} &&\!\!\!\!\! P_S\min\left(
 \frac{1}{\Ssr}\|\qh_{SR}\|^2  -  \frac{ \frac{\kappa P_S}{\Ssr}\|\qh_{SR}\|^2 \|\qh_{SR}^\dag\qH_{RR}\qw_t\|^2 }
{ 1  +   \frac{\kappa P_S}{\Ssr}\|\qh_{SR}\|^2 \|\qH_{RR}\qw_t\|^2 }, \right.\nonumber\\
&&\left. \qquad\qquad\qquad\frac{\kappa P_S}{\Ssr\Srd} \|\qh_{SR}\|^2|
\qh_{RD}\qw_t|^2\right)\nonumber\\
\mbox{s.t.} &&\!\!\!\!\! \|\qw_t\|=1,
\eea
which is still difficult to solve. Therefore, instead of~\eqref{eqn:opt:with opt w_r} we solve the following problem by introducing an auxiliary variable $t$, as
\vspace{-0.4em}
\bea\label{eqn: nonconvex quadratic optimization}
f(t)\triangleq \max_{\|\qw_t\|=1} && |
\qh_{RD} \qw_t |^2\\
\mbox{s.t.} && \frac{  \frac{\kappa P_S}{\Ssr}\|\qh_{SR}\|^2 \|\qh_{SR}^\dag\qH_{RR}\qw_t\|^2 }{
1  +
 \frac{\kappa P_S}{\Ssr}\|\qh_{SR}\|^2\|\qH_{RR}\qw_t\|^2 } = t.\nonumber \eea
This is a nonconvex quadratic optimization problem with quadratic equality constraint. To solve the problem in~\eqref{eqn: nonconvex quadratic optimization}, we apply a similar approach as in~\cite{Zheng:JSPL2013} to convert the optimization problem to
\vspace{-0.3em}
\begin{align}\label{eqn: SDR prob}
\max_{\qW_t} &~ \mathsf{tr} (\qh_{RD}^\dag\qW_t\qh_{RD})\\
\mbox{s.t.} &~  \mathsf{tr} (\!\qW_t(\qH_{RR}^\dag\qh_{SR}\qh_{SR}^\dag \qH_{RR}
\!-\! \qH_{RR}\qH_{RR}^\dag)) \!\!=\! \frac{t\Ssr}{\kappa P_S \|\qh_{SR}\|^2}
\nonumber\\
&~\mathsf{tr}(\qW_t)=1,\nonumber
\end{align}
where ${\qW}_{t} = \qw_t\qw_t^{\dag}$  is a symmetric, positive semi-definite matrix. In order to solve~\eqref{eqn: SDR prob}, we can resort to the widely used semidefinite relaxation approach. By dropping the rank-1 constraint, the resulting problem becomes a semidefinite program, whose solution ${\qW}_{t}$ can be found by using the method provided in~\cite[Appendix B]{Zheng:JSPL2013}.

Denoting the optimal objective value of~\eqref{eqn: SDR prob} as $f(t)$, the SINR maximization problem can be formulated as
\vspace{-0.5em}
\begin{align}\label{eq:sum rate maximization problem one dimentional}
\underset{t\geq 0}{\text{max}}&
~ \min\left( P_S\left( \frac{1}{\Ssr}\|\qh_{SR}\|^2  \!-\! t\right),\!  \frac{\kappa P_S}{\Ssr\Srd} \|\qh_{SR}\|^2 f(t)\right).
\end{align}
Therefore, in order to solve~\eqref{eqn:opt}, it remains to perform a one-dimensional optimization with respect to the variable $t$.

 \vspace{-0.2em}
\subsection{Transmit ZF (TZF) Scheme}
In the transmit ZF scheme, the FD relay takes advantage of the multiple transmit antennas to completely cancel the LI. To ensure this is feasible, the number of the transmit antennas at $R$ should be greater than one, i.e., $M_T>1$. In this case $\qw_r =
\frac{\qh_{SR}^{\dag}}{\|\qh_{SR}\|}$ and $\qw_t$ is the solution of the following optimization problem:
\vspace{-0.3em}
   \bea\label{eqn:wt}
    \max_{\qw_t} &&  |\qh_{RD}\qw_t |^2\\
    \mbox{s.t.} && \|\qw_t\|=1,\notag\quad \qh_{SR}^\dag\qH_{RR}\qw_t =0. \notag
 \eea
From the ZF constraint, we know that $\qw_t$ lies in the null space of $\qh_{SR}^\dag\qH_{RR}$. Denoting $\qB\triangleq \qI -\frac{\qH_{RR}^\dag\qh_{SR}\qh_{SR}^\dag\qH_{RR}}{\|\qh_{SR}^\dag\qH_{RR}\|^2}$, we have $\qw_t^{\ZF} = \frac{\qB \qh_{RD}^{\dag}}{\|\qB\qh_{RD}^{\dag}\|}$.

\vspace{-0.1em}
\subsection{Receive ZF (RZF) Scheme}
As an alternative solution, the transmit beamforming vector can be set using the MRT principle, i.e., $\qw_t =\frac{\qh_{RD}^{\dag}}{\|\qh_{RD} \|}$, and $\qw_r$ based on the ZF criterion. To ensure feasibility of ZF, $R$ should equipped with $M_R>1$ receive antennas.

By using similar procedure as shown for the transmit ZF scheme, the combining vector $\qw_r$ can be obtained as $\qw_r^{\ZF} = \frac{\qD \qh_{SR}}{\|\qD    \qh_{SR}\|}$ with $\qD\triangleq  \qI -
\frac{\qH_{RR}\qh_{RD}^\dag\qh_{RD} \qH_{RR}^\dag}{\|\qH_{RR}\qh_{RD}\|^2}$.

\subsection{MRC/MRT Scheme}
 For the MRC/MRT scheme, $\qw_r$ and $\qw_t$ are set to match the first hop and second hop channel, respectively. Hence, $\qw_r^{\MRC}  = \frac{\qh_{SR}^{\dag}}{\| \qh_{SR} \|}$ and $\qw_t^{\MRT} =\frac{\qh_{RD}^{\dag}}{\|\qh_{RD} \|}$.

\section{Outage Probability}
In this section, we investigate the outage probability of the considered FD relay system assuming TZF, RZF, and MRC/MRT schemes. In case of the  optimal scheme, derivation of the outage is difficult and we use simulations in Section~\ref{sec:numerical results}.

The outage probability is an important performance metric, which is defined as the probability that the instantaneous SINR falls below a predefined threshold, $\gamth$. Mathematically, it can be written as
\vspace{-0.3em}
\begin{align}\label{eq:OP definition}
\Pout = \Prob(\gamma < \gamth) = F_{\gamma}(\gamth).
\end{align}

\subsection{TZF Scheme}
By substituting the $\qw_t^{\ZF}$ and $\qw_r^{\MRC}$ into~\eqref{eq: e2e snr general}, the end-to-end SNR $\gamma_{\TZF}$ can be expressed as
\vspace{-0.3em}
\begin{align}\label{eq:gammaTZF}
\gamma_{\TZF}
&=\frac{P_S}{(1-\alpha) d_1^{\tau}} \|\hSR\|^2 \min\left(1\!-\alpha, \frac{\eta\alpha}{d_2^{\tau}}\|\thRD\|^2\right),
\end{align}
where $\thRD$ is an $(M_T-1) \times 1$ vector with $\|\thRD\|^2 \sim \chi_{2(M_T-1)}^2$. Let $Y = \min\left(1-\alpha,  \frac{\eta\alpha}{d_2^{\tau}}\|\thRD\|^2 \right)$. Then invoking the cumulative density function (cdf) of $Y$ presented in~\cite[Appendix II]{Zhu:TCOM:2015}, the cdf of $\gamma_{\TZF}$ can be obtained as
\vspace{-0.3em}
\begin{align}
\label{eq:cdf of gammaTZF}
&F_{\gamma_{\TZF}} (z)
  =\!
  1 \!-\! \int_{\frac{ d_1^{\tau}z }{\rho_1}}^{\infty}\!\!
  \frac{Q\left(M_T\!-1,\frac{d_1^\tau d_2^\tau}{\kappa \rho_1} \frac{z}{x}\right)}
  {\Gamma(M_R)}
   x^{M_R\!-1}e^{-x} dx.
\end{align}
where $Q(a,x) = \Gamma(a,x)/\Gamma(a)$ and $\Gamma(a,x)$ is the upper incomplete Gamma function~\cite[Eq. (8.350.2)]{Integral:Series:Ryzhik:1992}, and $\rho_1 = P_S$ is the SNR of the first hop.

To the best of the authors's knowledge, the integral in~\eqref{eq:cdf of gammaTZF} does not admit a closed-form expression. However, this single integral expression can be efficiently evaluated numerically using software such as Matlab or Mathematica.

To gain further insights, we now look into the high SNR regime and derive a simple approximation for the outage probability, which enables the characterization of the achievable diversity order of the TZF scheme.
\begin{proposition}\label{Propos:hsnr:TZF}
In the high SNR regime, i.e., $\rho_1 \rightarrow \infty$, the outage probability of FD relaying with the TZF scheme can be approximated as~\eqref{eq:cdf of gammaTZF: Asymptotic} at the top of the next page, where $\psi(x)$ is the digamma function~\cite[Eq. (6.3.1)]{Abramowitz_Handbook_1970}.

\begin{figure*}
\begin{align}\label{eq:cdf of gammaTZF: Asymptotic}
F_{\gamma_{\TZF}^{\infty}}(z)& \approx
 \begin{cases}
 \left(\frac{1}{\Gamma(M_R+1)}
 +
  \frac{1}{\Gamma(M_T-1)\Gamma(M_R)}
  \!\!\!\!\!\!\sum\limits_{\substack{k=0 \\
       k\neq M_R - M_T +1}}^{\infty}\!\!\!\!\!\!\!
 \frac{(-1)^{k+1}}{k!(k+M_T)}
  \left(\frac{d_2^{\tau}}{\kappa} \right)^{M_T+k-1}
  \!\!\!\frac{1}{M_R - M_T-k +1}
   \right)
 \left(\bGsd\right)^{M_R}\!\!\!\!, &                                M_T>M_R+\!1,  \\
\frac{1}{\Gamma(M_R+1)}\left(1+\frac{1}{\Gamma(M_R)}
        \left(\ln\left({\rho_1}\right)-
        \ln\left({ d_1^{\tau}z }\right) + \psi(1)\right) \left(\frac{ d_2^{\tau}}{\kappa} \right)^{M_R}  \right)
        \left(\bGsd\right)^{M_R},
         & M_T = M_R+\!1\\
        \frac{\Gamma(M_R - M_T + 1)}{\Gamma(M_T)\Gamma(M_R)}\left(\frac{d_2^{\tau}}{\kappa} \right)^{M_T-1}
        \left(\bGsd\right)^{M_T -1}, &  M_T < M_R+\!1
\end{cases}
\end{align}
\hrule
\end{figure*}

\end{proposition}
\begin{proof}
Due to limited space proof is omitted.
\end{proof}

By inspecting~\eqref{eq:cdf of gammaTZF: Asymptotic}, we see that the TZF scheme achieves a diversity order of $\min(M_R, M_T-1)$. Moreover, we notice that for $M_R+1=M_T$, $F_{\gamma_{\TZF}^{\infty}}(z)$ decays as $\rho_1^{-M_R}\ln(\rho_1)$ rather than $\rho_1^{-M_R}$ as in the conventional case which implies that in the energy harvesting case the slope of $F_{\gamma_{\TZF}^{\infty}}(z)$ converges much slower compared with that in the constant power case.
\vspace{-0.3em}
\subsection{RZF Scheme}
Invoking~\eqref{eq: e2e snr general}, and using $\qw_r^{\ZF}$ and $\qw_t^{\MRT}$, the end-to-end SNR $\gamma_{\RZF}$ can be expressed as
\vspace{-0.4em}
\begin{align}\label{eq:gammaRZF}
\gamma_{\RZF} = Z P_S \min  \left(X_1,\frac{\kappa}{ d_2^{\tau}}
\|\hRD\|^2\right),
\end{align}
where $Z \triangleq\frac{1}{d_1^{\tau}}\left( \|\thSR\|^2  + |\tilde{h}_1|^2\right)$ and $X_1 = \frac{\|\thSR\|^2}{\|\thSR\|^2  + |\tilde{h}_1|^2}$. It is well known that $Z$ follows central chi-square distribution with $2M_R$ degrees-of-freedom, denoted as $Z \sim \chi^2_{2M_R}$ and that $X_1$ follows a beta distribution with shape parameters $M_R-1$ and $1$, denoted as $X_1 \sim \mathsf{Beta} (M_R-1, 1)$ with $ F_{X_1}(x)\!=\!x^{M_R-1},~0\!<x<1$,~\cite[p. 138]{MathematicalStatistics_1978}.

Moreover, by denoting $Y_1 = \frac{\kappa}{d_2^{\tau}}\|\hRD\|^2 $, we have $F_{Y_1}(y)=P\left(M_T, \frac{d_2^{\tau}}{\kappa} y \right)$, where $P(a,x)=\gamma(a,x)/\Gamma(a)$ is the regularized lower incomplete Gamma function~\cite[Eq. (6.5.1)]{Abramowitz_Handbook_1970}. With $F_{X_1}(x)$ and $F_{Y_1}(y)$ in hand, the cdf of $\gamma_{\RZF}$ can be expressed as~\eqref{eq:cdf of gammaRZF}. Although Eq.~\eqref{eq:cdf of gammaRZF} does not admit a closed-form solution, it can be efficiently evaluated.
\vspace{-0.0em}
\begin{figure*}
\begin{align}\label{eq:cdf of gammaRZF}
F_{\gamma_{\RZF}}(z)
& \!= 1\!-\! Q\left(\!M_R, \bGsd\!\right)+\frac{1}{\Gamma(M_R)}
\left( \!\int_{\bGsd}^{\infty}\!\!
P\left(\!M_T, \frac{d_1^{\tau}d_2^{\tau}}{\kappa \rho_1} \frac{z}{x} \!\right)
x^{M_R-1}e^{-x} dx
\!+\! \left(\!\bGsd\!\right)^{M_R-1}\!\!
\!\!\int_{\bGsd}^{\infty}\!\!
Q\left(\!M_T, \frac{d_1^{\tau}d_2^{\tau}}{\kappa \rho_1} \frac{z}{x} \!\right)
e^{-x} dx \right).
\end{align}
\hrule
\end{figure*}
Now, we look into the high SNR regime, and investigate the diversity order achieved by this scheme.
\begin{proposition}\label{Propos:hsnr:RZF}
In the high SNR region, i.e., $\rho_1 \rightarrow \infty$, the outage probability of the FD relaying system
with the RZF scheme can be approximated as
\vspace{-0.0em}
\begin{align}
F_{\gamma_{\RZF}^{\infty}}(z)& \approx
 \begin{cases}
 \frac{1}{\Gamma(M_R)} \left(\bGsd\right)^{M_R-1},\qquad\quad\quad~ M_R < M_T+\!1,\\
 \frac{1}{\Gamma(M_R)}\left(1 \!+\! \frac{1}{\Gamma(M_T)} \left(\frac{d_2^\tau}{\kappa}\right)^{M_T}\right)
 \left(\frac{d_1^{\tau}z}{\rho_1}\right)^{M_T} ,
 & \\ \qquad\qquad\qquad\qquad\qquad\qquad\quad  M_R = M_T+\!1,\\
 \frac{\Gamma(M_R-\!M_T)}{\Gamma(M_R)\Gamma(M_T+1)}\!\!
 \left(\frac{d_2^\tau}{\kappa}\right)^{M_T}\!\!\!\!
 \left(\bGsd\!\right)^{M_T}\!\!,M_R> M_T+\!1.
\end{cases}\nonumber
\end{align}
\end{proposition}
\begin{proof}
The proof is omitted due to space limit.
\end{proof}

Proposition~\ref{Propos:hsnr:RZF} indicates that the RZF scheme achieves
a diversity order of $\min(M_R-1, M_T)$.
\vspace{-0.3em}
\subsection{The MRC/MRT Scheme}
Substituting $\qw_r^{\MRC}$ and $\qw_t^{\MRT}$ into~\eqref{eq: e2e snr general}, analysis of the MRC/MRT scheme for arbitrary $M_T$ and $M_R$ appears to be cumbersome. Therefore, in the sequel we consider two special cases as follows:

\textbf{Case-1)} $M_T=1$, $M_R\geq 1$: In this case $\frac{|\hSR^{\dag}\vH_{R R}\hRD^{\dag}|^2}{\|\hSR \hRD\|^2}$ is given by
\vspace{-0.3em}
\begin{align}\label{eq: WRHaaWT: MRC/MRT case-2}
X_1 &\triangleq |\qw_r^{\MRC} \vh_{R R}|^2
= \frac{ \mid \hat{\vh}_{SR,1}\mid^2 }{ \|\hSR \|^2}  \| \vh_{RR}\|^2.
\end{align}
For notational convenience, we define  $c_1 = \frac{P_S}{ d_1^{\tau}}$,  $c_2 = \frac{P_S \kappa\Sap}{ d_1^{\tau}}$, $c_3 =\frac{P_S\kappa}{ d_1^{\tau}d_2^{\tau}} $ and write the end-to-end SINR as
\vspace{-0.3em}
\begin{align}
\gamma_{\MRC} &=\!
\min\left(\!   \frac{c_1 \|\hSR\|^2}
{c_2|\hat{\vh}_{SR,1}|^2  \| \vh_{RR}\|^2\! +\! 1},
c_3 \|\hSR\|^2 |h_{RD}|^2\!\right).\nonumber
\end{align}
Let us denote $X  = c_1/\left(c_2 X_1  +\frac{1}{Y_1}\right)$ and $Y  = c_3Y_1Y_2$, with $Y_1=\|\hSR\|^2$ and $Y_2 = |h_{RD}|^2$. Conditioned on $Y_1$, the random variables (RVs) $X$ and $Y$ are independent and hence
\begin{align}\label{eq:eq:OP exact case-1 joint cdf of X1 and X2}
F_{\gamma^{\MRC}}(z)
 &=1-\int_{\frac{z}{c_1}}^{\infty}
 F_{X_1} \left(\frac{1}{c_2}\left(\frac{c_1}{z} -\frac{1}{y}\right)\right)\nonumber\\
 &\qquad\qquad\times
 \left(1-F_{Y_2}\left(\frac{z}{c_3 y}\right)\right)f_{Y_1}(y) dy.
\end{align}
In order to evaluate~\eqref{eq:eq:OP exact case-1 joint cdf of X1 and X2} we require the cdf of the RV, $X_1$. Note that $ \|\vh_{RR}\|^2 \sim\!\chi_{2M_R}^2$, $ Z\!\triangleq\frac{ \mid \hat{\vh}_{SR,1}\mid^2 }{\|\hSR \|^2}$ is distributed as $Z\sim \Bta(1,M_R\!-\!1)$, and the cdf of $X_1$ can be evaluated as
\begin{align}\label{eq:cdf of Y:chi2 times beta}
F_{X_1}(t)
&=G_{2 3}^{2 1} \left(t \  \Big\vert \  {1, M_R \atop 1, M_R, 0} \right),
\end{align}
where \small{$G_{p q}^{m n} \left( z \  \vert \  {a_1\cdots a_p \atop b_1\cdots b_q} \right)$ }\normalsize is the Meijer G-function~\cite[ Eq. (9.301)]{Integral:Series:Ryzhik:1992}.
Now, using the cdf of RV $Y_2$, and substituting~\eqref{eq:cdf of Y:chi2 times beta} into~\eqref{eq:eq:OP exact case-1 joint cdf of X1 and X2} we obtain
\begin{align}\label{eq:OP exact case-1 final}
F_{\gamma_{\MRC}}(z)
 &= 1-\! \frac{1}{\Gamma(M_R)}
 \!\int_{\frac{z}{c_1}}^{\infty}\!
G_{2 3}^{2 1} \left( \frac{1}{c_2}\left(\frac{c_1}{z}\!- \!\frac{1}{y}\right) \  \Big\vert \  {1, M_R \atop 1, M_R, 0} \right)
\nonumber\\
&\qquad\qquad\times y^{M_R-1} e^{-\left(y+\frac{z}{c_2 y}\right)} dy.
\end{align}
To the best of the authors' knowledge, the integral in~\eqref{eq:OP exact case-1 final} does not admit a closed-form
solution. However, \eqref{eq:OP exact case-1 final} can be evaluated numerically.

\emph{\textbf{Case-2)}} $M_R=1$, $M_T\geq 1$: In this case $\frac{|\hSR^{\dag}\vH_{R R}\hRD^{\dag}|^2}{\|\hSR \hRD\|^2}$ simplifies to
\begin{align}\label{eq: WRHaaWT: MRC/MRT case-1}
Y_1&= (\vh_{R R} \boldsymbol{\Phi}_t
\diag\{1,0,\cdots,0\} \boldsymbol{\Phi}^{\dag}_t\vh_{R R}^{\dag} )= |\hat{\vh}_{RR,1}|^2,
\end{align}
where $\hat{\vh}_{RR} =\vh_{RR} \boldsymbol{\Phi}_t$. Hence, $\gamma_{\MRC}$ can be written as
\begin{align}
\gamma_{\MRC} &=
\min\left(   \frac{c_1 |h_{SR}|^2}
{c_2
|h_{SR}|^2 |\hat{\vh}_{RR,1}|^2 + 1}, c_3|h_{SR}|^2 \|\hRD\|^2\right).\nonumber
\end{align}
Let us define $X  = c_1 X_1/\left(c_2 X_1 Y_1 + 1\right)$, and $Y = c_3X_1Y_2$ where $X_1=|h_{SR}|^2$, $Y_2=\|\hRD\|^2$. Note that conditioned on $X_1$, the RVs, $X$ and $Y$ are independent and hence we have
\begin{align}
&F_{\gamma_{\MRC}}(z)
=1\!-\! \int_{\frac{z}{c_1}}^{\infty} \!\!\left(1\!-\! F_{X| X_1} (z)\right)\left(1\!-\!F_{Y|X_1} (z)\right) f_{X_1}(x) dx,\nonumber\\
&~=1\!-\!\int_{\frac{z}{c_1}}^{\infty}\!\!
 \left(1\!-\!e^{-\frac{1}{c_2x} \left(\frac{c_1 x}{z}-1\right)}\right)
 Q\left(M_T,\!\frac{z}{c_3 x}\!\right)
  e^{-x} dx.
\end{align}

\vspace{-1em}
\subsection{Half-Duplex Scheme}
We now present the outage probability of the HD relaying system, which serves as a benchmark for
performance comparison. The energy harvesting phase of the HD relaying is the same as that of the FD
relaying system. However, in the information transmission phase, the remaining $(1-\alpha)$ portion of block time is equally partitioned
into two time slots for source and relay transmissions~\cite{Nasir:TWC:2013}. The end-to-end SNR of the HD relaying scheme
can be computed as
\begin{align}\label{eq:gammaTZF}
\gamma_{\HD} = \|\hSR\|^2\min  \left(c_1 , 2c_3 \|\hRD\|^2 \right).
\end{align}
The required cdf can be obtained by replacing $M_T-1$ with $M_T$ in~\eqref{eq:cdf of gammaTZF}, and is given by
\begin{align}\label{eq:cdf for gamma}
&F_{\gamma_{\HD}}(z)  = 1- \int_{\frac{ d_1^{\tau}z}{P_S}}^{\infty}\frac{Q\left(M_T, \frac{d_1^{\tau}d_2^{\tau}}{2\kappa\rho_1} \frac{z}{x}\right)}{ \Gamma(M_R)}
     x^{M_R-1}e^{-x} dx.
\end{align}
\section{Numerical and Simulation Results }\label{sec:numerical results}
\begin{figure}[t]
\centering
\includegraphics[width=80mm, height=62mm]{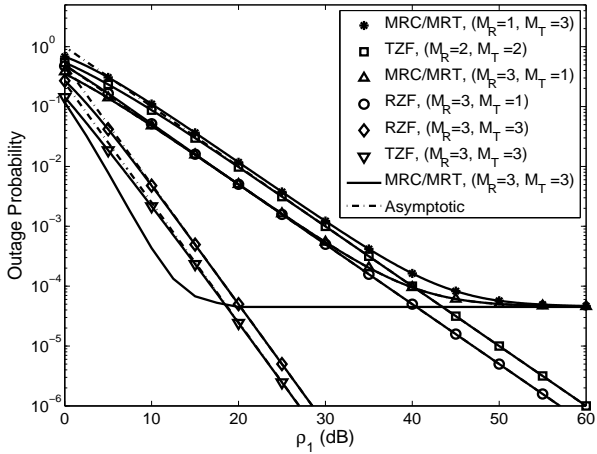}
\vspace{-0.8em}
\caption{Outage probability versus $\rho_1$ of the TZF, RZF and MRC/MRT schemes for different antenna conﬁgurations. ($\gamma_{\mathsf{th}}=0$ dB, and $\alpha=0.5$).}
\label{fig: Outage probability}
\vspace{-0.8em}
\end{figure}
\begin{figure}[t]
\centering
\includegraphics[width=80mm, height=60mm]{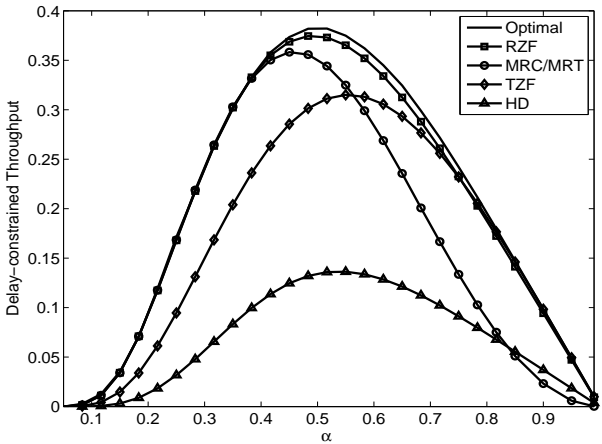}
\vspace{-0.8em}
\caption{Delay-constrained throughput of the proposed precoding schemes ($M_T=M_R=4$, $\rho_1=10$ dB, $\Sap =0.3$, $d_1=d_2=2$, $\tau=3.1$ and $R_c=1$ bit/sec/Hz). Results for the MRC/MRT scheme are from simulations.}
\label{fig: Delay-constraint-throu}
\vspace{-1.4em}
\end{figure}
We now present numerical and simulation results to investigate the impact of key system parameters on the performance. In all cases, we have set $\eta=1$.

Fig.~\ref{fig: Outage probability} compares the outage probability of the considered schemes with different antenna configurations and for $\alpha=0.5$, $d_1=d_2 =1$, $\Sap=0.1$, and $\gamma_{\mathsf{th}}=0$ dB. When the transmit power is high and $\alpha$ remains fixed, an excessive amount of energy will be collected at the relay, which is detrimental for the MRC/MRT scheme since it results in strong LI. Therefore, the outage performance of the MRC/MRT scheme exhibits an outage floor at high SNRs. On the contrary, the outage probability of the TZF and RZF precoding schemes decays proportional to the diversity orders provided in Proposition~\ref{Propos:hsnr:TZF} and~\ref{Propos:hsnr:RZF}, respectively since the relay is capable of canceling LI. However, a proper choice of $\alpha$ can improve the outage of the MRC/MRT scheme to some level at high SNR. On the other hand observation of different curves in the low SNR region reveal that the MRC/MRT scheme outperforms the TZF and RZF schemes even at high LI level. This is because at low SNR, overall interference can be treated as noise and therefore MRC filtering helps to maximize the SNR. Comparing the TZF and RZF schemes with the same diversity orders and different receive antenna numbers (i.e., TZF, with $M_R =2$ and $M_T=2$, and RZF  with $M_R =3$ and $M_T=1$) we see that the additional receive antenna could harvest more energy to facilitate information transfer. Moreover, for the case where $M_R=M_T$, TZF achieves a higher array gain.

Fig.~\ref{fig: Delay-constraint-throu} shows the impact of optimal time split $\alpha$ on the delay-constrained transmission throughput defined as: $R (\alpha) =\theta(1-\Pout)R_c(1-\alpha)$, where $\theta =1$ for FD and $\theta =0.5$ for HD and the optimal time portion $\alpha$ can be obtained from
\vspace{-0.3em}
\begin{align*}
\alpha^{*}  = \argmax_{\alpha} R (\alpha), \quad  0<\alpha<1.
\end{align*}
As expected, the optimal scheme exhibits the best throughput out of all precoding schemes studied. The superior performance of the optimal scheme is more pronounced especially between $0.4$ and $0.8$ values of $\alpha$.
The highest throughput with optimized $\alpha$ for the optimal, RZF, MRC/MRT and TZF schemes are given by $0.382$, $0.374$, $0.358$ and $0.315$, respectively. Moreover, we see that each one of TZF, RZF and MRC/MRT precoder designs can surpass the others depending on the value of $\alpha$. Finally, all schemes achieve significant throughput gains as compared to the HD mode.
\section{Conclusion}\label{sec:conclusion}
In this paper, we have studied the outage probability and throughput of FD MIMO relaying with RF energy transfer. We designed the optimal precoder/decoder as well as investigated several sub-optimal low complexity precoding schemes. The MRC/MRT scheme provides a better outage performance at low-to-medium SNR, while the ZF precoders outperform the former at high SNR. Further, it is demonstrated that all proposed FD precoding schemes attain significant throughput gains compared to the HD mode. Therefore, FD relaying is a promising solution for implementing future
RF energy harvesting cooperative communication systems.

\bibliographystyle{IEEEtran}


\end{document}